\documentclass[reqno,11pt]{article} 
\usepackage[margin=1in]{geometry}
\usepackage{amsmath,amsthm,amssymb,amscd,amstext,amsfonts}
\usepackage{mathtools}
\usepackage{mathrsfs}
\usepackage{thmtools}
\usepackage{hyperref}
\usepackage{url}

\newtheorem{theorem}{Theorem}

\newtheorem{proposition}{Proposition}

\theoremstyle{remark}

\newcommand{\R}{\mathbb{R}}
\newcommand{\norm}[1]{\left\Vert#1\right\Vert}  
\newcommand{\cara}{Carath{\'e}odory}

\newcommand{\convex}{conv}

\begin{document}

\title{Newman's theorem via Carath{\'e}odory}

\author{
Yaqiao Li\footnote{Shenzhen University of Advanced Technology, Shenzhen, China, liyaqiao@suat-sz.edu.cn},
Ali Mohammad Lavasani\footnote{Concordia University, ali.mohammadlavasani@concordia.ca},
Mehran Shakerinava\footnote{McGill University, mehran.shakerinava@mail.mcgill.ca}}

\maketitle

\begin{abstract}
    We give a short geometric
    proof of Newman's theorem in communication complexity by applying the classical and approximate Carath{\'e}odory theorems from convex geometry.
\end{abstract}

Newman's theorem is a fundamental result in communication complexity. It says that using public randomness can only gain a logarithmic advantage in communication cost comparing to private randomness. The standard proof \cite{KN,roughgarden2016communication,rao2020communication} uses the probabilistic method and the Chernoff bound. Here, we apply Carath{\'e}odory's theorems directly. Our proof is not too different from the standard proof, as the approximate Carath{\'e}odory is proved via a similar method. Though, the alternative proof seems neater. Besides, we note that different versions of approximate Carath{\'e}odory theorems and their various applications in computer science have been discussed in \cite{Cara_approx_1,Cara_approx_2}. Hence, presenting such a new proof may induce new applications.

Let $\convex(S)$ denote the convex hull of a set $S \subseteq \R^n$. Let $\norm{\cdot}_\infty$ denote the usual $L_\infty$ norm. Let $R^{pri}_{\epsilon}$ and $R^{pub}_{\epsilon}$ denote the communication cost using private and public randomness, respectively.

\begin{proposition}[\cara,~\cite{Cara,Cara2}]
    \label{prop:Cara}
     Let $S \subseteq \R^n$. Every $x \in \convex(S)$ can be written as a convex combination of at most $n+1$ points from $S$. 
\end{proposition}

\begin{proposition}[Approximate \cara,~\cite{Cara_approx_1}]
    \label{prop:Cara_approx}
    Let $\delta > 0$. Let $S\subseteq \R^n$ satisfy
        $\max_{y \in S} \norm{y}_\infty \le 1$.
    For every $x \in \convex(S)$, there exists $x' \in \convex(S)$, such that
        $\norm{x - x'}_\infty \le \delta$,
    and $x'$ is a convex combination of at most $O(\frac{\log |S|}{\delta})$ many points from $S$.
\end{proposition}

\begin{theorem}[Newman,~\cite{Newman}]
    Let $f: \{0,1\}^n \times \{0,1\}^n  \to \{0,1\}$ be a Boolean function. Let $\epsilon, \delta > 0$. Then,
        $R^{pri}_{\epsilon + \delta}(f) \le R^{pub}_{\epsilon}(f) + O(\log n + \log \delta^{-1})$.
\end{theorem}

\begin{proof}
    We view $f \in \R^N$ for $N=2^{2n}$. By definition of $R^{pub}_{\epsilon}$, there exists a set of deterministic protocols, denoted by 
        $S = \{P_1, \ldots, P_q\}$,
    where $P_i$ is a Boolean function and $P_i \in \R^N$,  satisfying
    \begin{equation}    \label{eq:approx_1}
        \norm{f-P}_\infty \le \epsilon
    \end{equation}
    for some $P \in \R^N$ where $P$ is a probabilistic average of $P_i$'s, hence, $P \in \convex(S)$. By Proposition \ref{prop:Cara}, there exists a subset 
        $S' \subseteq S$
    of size $|S'| \le N+1$ such that
        $P \in \convex(S')$.
    Because every $P_i$ is a Boolean-valued function, we have
        $\max_{y \in S'} \norm{y}_\infty \le 1$.
    By Proposition \ref{prop:Cara_approx}, there exists another point
        $P' \in \convex(S')$
    that can be written as a convex combination of at most $k = O(\frac{\log |S'|}{\delta})$ points from $S'$, and 
    \begin{equation} \label{eq:approx_2}
        \norm{P-P'}_\infty \le \delta.
    \end{equation}
    WLOG, let 
        $P' = \sum_{i=1}^k \mu_i P_i$,
    where
        $\sum_i \mu_i = 1$
        and
        $\mu_i \in [0,1]$.
    Now, Alice and Bob compute $f$ using private randomness, as follows: Alice privately samples a random number
        $i$ from the set $[k]$
    according to the distribution
        $(\mu_1, \ldots, \mu_k)$,
    she sends $i$ to Bob using at most
        $\log k + 1 = O(\log\log |S'| + \log \delta^{-1}) = O(\log n + \log \delta^{-1})$
    bits. Then, they run the deterministic protocol $P_i$.  By \eqref{eq:approx_1} and \eqref{eq:approx_2},
        $\norm{f-P'}_\infty \le \epsilon + \delta$.
    This gives a private randomness protocol with the desired error and cost.
\end{proof}

\bibliography{mybib}{}
\bibliographystyle{plain}

\end{document}